\documentclass[12pt]{iopart}
\usepackage{amsfonts}
\usepackage{amssymb}

\providecommand{\U}[1]{\protect\rule{.1in}{.1in}}
\providecommand{\U}[1]{\protect\rule{.1in}{.1in}}

\newtheorem{acknowledgement}{Acknowledgement}

\newtheorem{definition}{Definition}

\newtheorem{lemma}{Lemma}

\newtheorem{proposition}{Proposition}
\newtheorem{remark}{Remark}

\newenvironment{proof}[1][Proof]{\noindent\textbf{#1.} }{\ \rule{0.5em}{0.5em}}
\def\dfrac{\frac}

\begin{document}

\title {Local and nonlocal solvable structures in ODEs reduction}

\author{D Catalano Ferraioli$^1$, P Morando$^2$}

\address{$^1$Dipartimento di Matematica, Universit\`{a} di Milano, via Saldini 50, I-20133 Milano, Italy}
\address{$^2$Istituto di Ingegneria Agraria, Facolt\`{a} di Agraria, Universit\`{a} di Milano, Via Celoria, 2 - 20133 Milano, Italy}

\ead{catalano@mat.unimi.it and paola.morando@polito.it}

\begin{abstract}

Solvable structures, likewise solvable algebras of local
symmetries, can be used to integrate scalar ODEs by quadratures.
Solvable structures, however, are particularly suitable for the
integration of ODEs with a lack of local symmetries. In fact,
under regularity assumptions, any given ODE always admits solvable
structures even though finding them in general could be a very
difficult task. In practice a noteworthy simplification may come
by computing solvable structures which are adapted to some
admitted symmetry algebra. In this paper we consider solvable
structures adapted to local and nonlocal symmetry algebras of any
order (i.e., classical and higher). In particular we introduce the
notion of nonlocal solvable structure.

\end{abstract}

\vspace{2pc} \noindent{\it Keywords}: Solvable structures, local
symmetries, nonlocal symmetries, $\lambda$-symmetries, ODE
reduction.


\section{Introduction}

In recent years a renewed interest in the classical works of Lie and Cartan
led to further developments in the geometrical methods for the study of
ordinary differential equations (ODEs). In particular, symmetry methods for
ODEs have received an increasing attention in the last fifteen years. It is
well known that if an ODE $\mathcal{E}$ admits a local symmetry, this can be
used to reduce the order of $\mathcal{E}$ by one. This procedure, usually
referred to as symmetry reduction method, is particularly useful when
$\mathcal{E}$ is a $k$-order ODE whose local symmetries form a solvable
$k$-dimensional Lie algebra. In this case, in fact, $\mathcal{E}$ can be
completely integrated by quadratures \cite{Fe,Olv1,Vin-et-al}. However,
finding all local symmetries of a given ODE $\mathcal{E}$ is not always
possible. In fact local symmetries (classical or higher) of a $k$-order ODE in
the unknown $u$ are defined by the solutions of a linear PDE depending on the
derivative of $u$ up to the order $k-1$. Since the general solution of this
PDE cannot be found unless one knows the general solution of $\mathcal{E}$,
one usually only searches for particular solutions depending on derivative of
$u$ up to order $k-2$. Therefore, in practice, $\mathcal{E}$ could not be
reduced by quadratures if it does not admit a solvable $k$-dimensional algebra
of such symmetries.

Nevertheless, one may encounter equations solvable by quadratures but with a
lack of local symmetries. Examples of this kind, well known in recent
literature \cite{Bluman-Kumei,Bluman-Reid,Ca,GMM,Govinder-Leach,MuRo}, seem to
prove that local symmetries are sometimes inadequate and raise the question of
whether a generalization of the notion of symmetry would lead to a more
effective reduction method. For this reason, in the last decades various
generalizations of the classical symmetry reduction method have been proposed.
Among these, in our opinion, a special attention is deserved by the notions of
$\lambda$\emph{-symmetry} on the one hand and \emph{solvable structure} on the
other hand. Both notions, introduced in \cite{MuRo} and \cite{Ba},
respectively, are suitable for a wide range of applications
\cite{BaPri,BaPri1,Ca,CicGaeMor,GaMo,GMM,Mo,Sac,ShePri} and will be the
starting point of our forthcoming discussion. In particular, we will provide a
more effective reduction method which interconnects both notions.

The relevance of $\lambda$-symmetries is due to the fact that many equations,
not possessing Lie point symmetries, admit $\lambda$-symmetries and these can
be used to reduce the order exactly as in the case of standard symmetries.
Despite their name, however, $\lambda$-symmetries are neither Lie point nor
higher symmetries. As shown in \cite{Ca}, $\lambda$-symmetries of an ODE
$\mathcal{E}$ can be interpreted as shadows of some nonlocal symmetries. In
practice it means that, by embedding $\mathcal{E}$ in a suitable system
$\mathcal{E}^{\prime}$ determined by the function $\lambda$, any $\lambda
$-symmetry of $\mathcal{E}$ can be recovered as a local symmetry of
$\mathcal{E}^{\prime}$. This interpretation of $\lambda$-symmetries has many
advantages. First of all, it makes possible to give a precise geometric
meaning to the $\lambda$-invariance property; second, it allows to reinterpret
$\lambda$-symmetry reduction method as a particular case of the standard
symmetry reduction.

For what concern solvable structures, they were introduced by Basarab-Horwath
in \cite{Ba} and further investigated in \cite{BaPri,BaPri1,HaAt,ShePri}. The
original purpose was the generalization of the classical result relating the
integrability by quadratures of involutive distributions to the knowledge of a
sufficiently large solvable symmetry algebra. The main point of this extension
is that the vector fields assigning a solvable structure do not need to be
symmetries, nevertheless they allow to implement the reduction procedure.
However, finding local solvable structures for a $k$-order ODE is not much
easier than finding $k$-dimensional solvable algebras of local symmetries, if
any. A noteworthy simplification, in general, may come by computing solvable
structures adapted to the admitted symmetry algebras, if any. It follows that
one certainly takes advantage of the presence of any kind of symmetry. In
particular, one can include the nonlocal symmetries corresponding to $\lambda$-symmetries.

In this paper we consider solvable structures adapted to local and nonlocal
symmetry algebras of any order (i.e., classical and higher), and we introduce
the notion of nonlocal solvable structure.

The paper is organized as follows. In section 2 we collect all notations and
basic facts we need on the geometry of differential equations. In particular
we recall the definitions of local and nonlocal symmetries and also the
interpretation of $\lambda$-symmetries as shadows of nonlocal symmetries. In
section 3, we recall the results of \cite{Ba} and \cite{BaPri} in a form
suitable to our further discussion on the integration of ODEs. Then, in
section 4, we discuss our reduction scheme of ODEs via (higher order) local
and nonlocal solvable structures. In particular, in this section, we discuss a
practical method for the computation of solvable structures. Finally, in
section 5 we collect some examples which provide insight into the application
of the proposed reduction scheme.

\section{Preliminaries}

Troughout the paper we assume that the reader is familiar with the geometry of
differential equations. Nevertheless, we collect here some notations and basic
facts we need in the paper. The reader is referred to
\cite{Bluman-Anco,Bluman-Kumei,Gae,Olv1,Olv2,Ste, Vin-et-al} for further details.

\subsection{ODEs as submanifolds of jet spaces}

Let $M$ and $E$ be smooth manifolds and $\pi:E\rightarrow M$ be a
$q$-dimensional bundle. We denote by $\pi_{k}:J^{k}(\pi)\rightarrow M$ the
$k$-order \emph{jet bundle} associated to $\pi$ and by $j_{k}(s)$ the
$k$-order \emph{jet prolongation} of a section $s$ of $\pi$. Since in this
paper we are only concerned with the case $\dim M=1$, we assume that $M$ and
$E$ have local coordinates $x$ and $(x,u^{1},...,u^{q})$, respectively.
Correspondingly, the induced natural coordinates on $J^{k}(\pi)$ will be
$(x,u_{i}^{a})$, $1\leq a\leq q$, $i=0,1,...,k$, where the $u_{i}^{a}$'s are
defined by $u_{i}^{a}(j_{\infty}(s))=d^{i}\left(  u^{a}(s)\right)  /dx^{i}$,
for any section $s$ of $\pi$. Moreover, when no confusion arises, Einstein
summation convention over repeated indices will be used.

The $k$\emph{-order jet space} $J^{k}(\pi)$ is a manifold equipped with the
smooth distribution $\mathcal{C}^{k}$ of tangent planes to graphs of $k$-order
jet prolongations $j_{k}(s)$. This is the \emph{contact (or Cartan)
distribution} \emph{of} $J^{k}\left(  \pi\right)  $.

In this framework a $k$-th order system of differential equations can be
regarded as a submanifold $\mathcal{E}\subset J^{k}(\pi)$ and any solution of
the system is a section of $\pi$ whose $k$-order prolongation is an integral
manifold of the restriction $\left.  \mathcal{C}^{k}\right\vert _{\mathcal{E}%
}$ of the contact distribution to $\mathcal{E}$. In this paper we will deal
only with (systems of) ordinary differential equations $\mathcal{E}$ which are
in \emph{normal form} and \emph{not underdetermined}.

The natural projections $\pi_{h,k}:J^{h}(\pi)\rightarrow J^{k}(\pi)$, for any
$h>k$, allow one to define the bundle of infinite jets $J^{\infty}%
(\pi)\rightarrow M$ as the inverse limit of the tower of projections
$M\longleftarrow E\longleftarrow J^{1}(\pi)\longleftarrow J^{2}(\pi
)\longleftarrow...$.

The manifold $J^{\infty}(\pi)$ is infinite dimensional with induced
coordinates $(x,u_{i}^{a})$, $1\leq a\leq q$, $i=0,1,...$, and the
$\mathbb{R}$-algebra of smooth functions on $J^{\infty}(\pi)$ is defined as
$\mathcal{F}(\pi)=\cup_{l}C^{\infty}(J^{l}(\pi))$. The set $D(\pi)$ of vector
fields on $J^{\infty}(\pi)$ has the structure of a Lie algebra, with respect
to the usual Lie bracket \cite{Chern,Olv1,Vin-et-al}. Moreover one can also
define a \emph{contact distribution }$\mathcal{C}$ on $J^{\infty}(\pi)$,
generated by the \emph{total derivative} operator
\[
D_{x}=\partial_{x}+u_{i+1}^{a}\partial_{u_{i}^{a}}.
\]
Notice that $\mathcal{C}$ is the annihilator space of the \emph{contact forms}
$\theta_{s}^{a}=du_{s}^{a}-u_{s+1}^{a}dx$.

Given a $k$-th order differential equation $\mathcal{E}=\{F=0\}$, the
$l$\emph{-th prolongation of }$\mathcal{E}$ is the submanifold $\mathcal{E}%
^{(l)}:=\left\{  D_{x}^{s}(F)=0:s=0,1,...,l\right\}  $. Analogously, we define
the \emph{infinite} \emph{prolongation} $\mathcal{E}^{(\infty)}$.

\subsection{Local symmetries}

The infinitesimal symmetries of $\mathcal{C}^{k}$ are called \emph{Lie
symmetries} of $J^{k}(\pi)$. These symmetries can be divided in two classes
(see \cite{Olv1,Vin-et-al}): \emph{Lie point symmetries}, which are obtained
by prolonging vector fields $X$ on $E$, and \emph{Lie contact symmetries},
which are obtained by prolonging symmetries $X$ on $J^{1}(\pi)$.

By definition, symmetries of $\mathcal{C}^{k}$ which are tangent to
$\mathcal{E}$ are called \emph{classical symmetries of }$\mathcal{E}$.

An analogous geometric picture holds on the infinite jet spaces. However, contrary to the case of finite order jet spaces, symmetries of
$\mathcal{C}$ cannot always be recovered by infinite prolongations of Lie
symmetries. In fact it can be proved that $X=\xi\partial_{x}+\eta_{i}%
^{a}\partial_{u_{i}^{a}}$ is an infinitesimal symmetry of $\mathcal{C}$ if and
only if $\xi,\eta^{a}\in\mathcal{F}(\pi)$ are arbitrary functions and
\[
\eta_{i}^{a}=D_{x}(\eta_{i-1}^{a})-u_{i}^{a}D_{x}(\xi),\qquad\eta_{0}^{a}%
=\eta^{a}.
\]
Hence, $X$ is the infinite prolongation of a Lie point (or contact) symmetry
iff $\xi,\eta_{0}^{a}$ are functions on $E$ (or $J^{1}(\pi)$, respectively).

On $J^{\infty}(\pi)$, symmetries $X$ of $\mathcal{C}$ which are tangent to
$\mathcal{E}^{(\infty)}$ are called \emph{higher symmetries} of $\mathcal{E}$
and are determined by the condition $\left.  X(F)\right\vert _{\mathcal{E}%
^{(\infty)}}=0$.

When working on infinite jets spaces, since $D_{x}$ is a trivial symmetry of
$\mathcal{C}$, it is convenient to gauge out from higher symmetries the terms
proportional to $D_{x}$. This leads to consider only symmetries in the so
called \emph{evolutive form}, i.e., symmetries of the form $X=D_{x}%
^{i}(\varphi^{a})\partial_{u_{i}^{a}}$, $\varphi^{a}:=\eta^{a}-u_{1}^{a}\xi$.
The functions $\varphi^{a}$ are called the \emph{generating functions} (or
characteristics) of $X$.

For the applications we are interested in, we only need to consider symmetries
of $\mathcal{C}$ which are tangent to $\mathcal{E}^{(\infty)}$. This choice
turns out to be also convenient since it noteworthy simplifies computations
(see \cite{Vin-et-al} and \cite{AKO} fore more details and other aspects of
$\infty$-jets theory).

As already remarked in the introduction, local symmetries (Lie or higher) of a
$k$-order ODE
\[
\mathcal{E}=\left\{  u_{k}=f(x,u,u_{1},...,u_{k-1})\right\}
\]
in the unknown $u$ are defined by the solutions of a linear PDE depending on
the derivative of $u$ up to the order $k-1$. This linear PDE is usually called
the determining equation and is obtained by writing in coordinates the
symmetry condition $\left.  X(u_{k}-f)\right|  _{\mathcal{E}^{(\infty)}}=0$.
Unfortunately, the general solution to this PDE cannot be found and one
usually searches only for particular solutions, if any, which depends on
derivative of $u$ up to order $k-2$. Hence, it is not unusual to run into ODEs
for which we are unable to compute local symmetries.

\subsection{Nonlocal symmetries and $\lambda$-symmetries}

In recent literature, to overcome the difficulties due to a lack of local
symmetries, various attempts have been made to find a more general notion of
symmetry. Hence, some new classes of symmetries have been introduced. Among
these generalizations there are those introduced by Muriel and Romero in
\cite{MuRo} and known as $\lambda$-symmetries.

The notion of \emph{nonlocal symmetry} we use in this paper furnishes the
conceptual framework to deal with many of these generalizations. We follow
here the approach to nonlocality based on the theory of \emph{coverings} (see
\cite{Kras-Vin} and also \cite{Vin-et-al}). However, since we only deal with
ODEs, our approach will be noteworthy simplified. The interested reader is
referred to \cite{Vin-et-al} for the general theory of coverings and nonlocal symmetries.

\begin{definition}
Let $\mathcal{E}$ be a $k$-order ODE on a $q$-dimensional bundle $\pi$. We
shall say that a smooth bundle $\kappa:\mathcal{\widetilde{E}\rightarrow
E}^{(\infty)}$ is a covering for $\mathcal{E}$ if the manifold
$\mathcal{\widetilde{E}}$ is equipped with a 1-dimensional distribution
$\mathcal{\widetilde{C}}=\{\mathcal{\widetilde{C}}_{p}\}_{_{p\in
\mathcal{\widetilde{E}}}}$ and, for any point $p\in\mathcal{\widetilde{E}}$,
the tangent mapping $\kappa_{\ast}$ gives an isomorphism between
$\mathcal{\widetilde{C}}$ and the restriction $\left.  \mathcal{C}\right\vert
_{\mathcal{E}^{(\infty)}}$ of the contact distribution of $J^{\infty}(\pi)$ to
$\mathcal{E}^{(\infty)}$.
\end{definition}

The dimension of the bundle $\kappa$ is called the dimension of the covering
and is denoted by $\dim(\kappa)$. Below we only consider the case $\dim(\kappa)=1$.

\begin{definition}
Nonlocal symmetries of $\mathcal{E}$ are the symmetries of the distribution
$\mathcal{\widetilde{C}}$ of a covering $\kappa:\mathcal{\widetilde
{E}\rightarrow E}^{(\infty)}$.
\end{definition}

Here follows a coordinate description of the notion of covering and nonlocal
symmetry. Since below we will restrict our attention only to nonlocal
symmetries of scalar ODEs, in the following formulas we will restrict only to
the case when $u^{a}=u$.

Let $\pi$ be a trivial $1$-dimensional bundle over $\mathbb{R}$, with standard
coordinates $(x,u)$, and let
\begin{equation}
\mathcal{E}:=\left\{  u_{k}=f(x,u,u_{1},...,u_{k-1})\right\}  .
\label{ODE_scalare}%
\end{equation}
Below we will consider only coverings where $\kappa$ is a trivial bundle
$\kappa:\mathcal{E}^{(\infty)}\times\mathbb{R}\rightarrow\mathcal{E}%
^{(\infty)}$, hence denoting by $w$ the standard coordinate in $\mathbb{R}$,
the distribution $\mathcal{\widetilde{C}}$ on $\mathcal{\widetilde{E}}$ is
generated by the vector field
\begin{equation}
\left.  \widetilde{D}_{x}\right|  _{\mathcal{\widetilde{E}}}=\overline{D}%
_{x}+H\partial_{w} \label{D_widetilde}%
\end{equation}
where $H$ is a smooth function on $\mathcal{E}^{(\infty)}\times\mathbb{R}$ and
$\overline{D}_{x}=\partial_{x}+u_{1}\partial_{u}+...+f\partial_{u_{k-1}}$ is
the restriction to $\mathcal{E}^{(\infty)}$ of the total derivative operator
on $J^{\infty}(\pi)$. Hence, the covering $\kappa$ is determined by the system (see \cite{Ca})
\begin{equation}
\mathcal{E}^{\prime}:=\left\{  u_{k}=f,w_{1}=H\right\}  . \label{covering}%
\end{equation}

Nonlocal symmetries of $\mathcal{E}$ are symmetries of the vector field
(\ref{D_widetilde}) and can be determined through a symmetry analysis of the
system $\left(  \mathcal{E}^{\prime}\right)  ^{(\infty)}$ on $J^{\infty
}(\widetilde{\pi})$. Therefore nonlocal symmetries of $\mathcal{E}$ have the
form
\begin{equation}
Y=\xi\partial_{x}+\eta_{i}\partial_{u_{i}}+\psi_{i}\partial_{w_{i}}
\label{pippo}%
\end{equation}
with $\eta_{i}=\widetilde{D}_{x}(\eta_{i-1})-\widetilde{D}_{x}(\xi)u_{i}$ and
$\psi_{i}=\widetilde{D}_{x}(\psi_{i-1})-w_{i}\widetilde{D}_{x}(\xi)$.

An interesting example of nonlocal symmetry occurring in literature is related
to the notion of $\lambda$-symmetry for an ODE $\mathcal{E}$ (see
\cite{MuRo}). Despite their name, in fact, $\lambda$-symmetries are neither
Lie symmetries nor higher symmetries of $\mathcal{E}$. Nevertheless, as
discussed in \cite{Ca}, $\lambda$-symmetries can be interpreted as shadows of
nonlocal symmetries. More precisely, $\mathcal{E}$ admits a $\lambda$-symmetry
$X$ iff $\mathcal{E}^{\prime}=\{u_{k}=f,w_{1}=\lambda\}$, with $\lambda\in
C^{\infty}(\mathcal{E})$, admits a (higher) symmetry with generating functions
of the form $\varphi^{\alpha}=e^{w}\varphi_{0}^{\alpha}(x,u,u_{1}%
,...,u_{k-1})$, $\alpha=1,2$.

Since $\lambda$-symmetries are of great interest in the applications, and
analogously their nonlocal counterparts, it is convenient to introduce the following

\begin{definition}
If a covering system $\mathcal{E}^{\prime}$, defined by (\ref{covering}) with
$H=\lambda$ and $\lambda\in C^{\infty}(\mathcal{E})$, admits a nonlocal
symmetry $Y$ with generating functions
\begin{equation}
\varphi^{\alpha}=e^{w}\varphi_{0}^{\alpha}(x,u,u_{1},...,u_{k-1}),\qquad
\alpha=1,2 \label{gen_func_lambda}%
\end{equation}
then $\mathcal{E}^{\prime}$ will be called a $\lambda$-covering for
$\mathcal{E}$ defined by (\ref{ODE_scalare}).
\end{definition}

\section{Solvable structures}

In this section we will recall basic definitions and facts on (local) solvable
structures in the form we need in our study. The reader is referred to
\cite{Ba} and \cite{BaPri} for further details.

The original purpose of solvable structures was to generalize the classical result stating
that the knowledge of a solvable $k$-dimensional algebra of symmetries for an
$(n-k)$-dimensional involutive distribution $\mathcal{D}$, on an
$n$-dimensional manifold $N$, guarantees that $\mathcal{D}$ can be integrated
by quadratures. Solvable structures are a generalization of solvable symmetry
algebras which allows one to keep this classical result.

Since in this paper we restrict our attention to the integration of
$1$-dimensional distributions, we will summarize the results of \cite{Ba} and
\cite{BaPri} only in this case.

Given a $1$-dimensional distribution $\mathcal{D}=<Z>$, on an $n$-dimensional
manifold $N$, the definition of a solvable structure for $\mathcal{D}$ is the following

\begin{definition}
\label{Def_solv}The vector fields $\{Y_{1},...,Y_{n-1}\}$ on $N$ form a
solvable structure for $\mathcal{D}=<Z>$ if and only if, denoting
$\mathcal{D}_{0}=\mathcal{D}$ and $\mathcal{D}_{h}=<Z,Y_{1},...,Y_{h}>$, the
following two conditions are satisfied:\newline\textbf{(i) }$\mathcal{D}%
_{n-1}=TN$;\newline\textbf{(ii)} $L_{Y_{h}}\mathcal{D}_{h-1}\subset
\mathcal{D}_{h-1}$, $\forall h\in\{1,...,n-1\}.$
\end{definition}

In particular, given a solvable structure $\{Y_{1},...,Y_{n-1}\}$ for
$\mathcal{D}$, one has the flag of integrable distributions%
\begin{equation}
<Z>=\mathcal{D}_{0}\subset\mathcal{D}_{1}\subset...\subset\mathcal{D}_{n-1}=TN
\label{flag}%
\end{equation}
with $\mathcal{D}_{s-1}=\{[X,Y]:X,Y\in\mathcal{D}_{s}\}$, and $[X,Y]$ denoting
the Lie bracket of $X$ and $Y$.

The main difference with the definition of a solvable symmetry algebra of
$\mathcal{D}$ is that the fields belonging to a solvable structure do not need
to be symmetries for $\mathcal{D}$. This, of course, gives more freedom in the
choice of the fields one can use in the integration of $\mathcal{D}$.

\begin{remark}
\label{Remark}It is straightforward, by the definition above, that in
principle a solvable structure for $\mathcal{D}$ always exists in a
neighborhood of a non-singular point for $\mathcal{D}$. In fact, if
$\{x^{i}\}$ is a local chart on $N$ such that $Z=\partial_{x^{1}}$, one can
simply consider the solvable structure generated by $\partial_{x^{j}}$,
$j=2,...,n$. This structure is in particular an Abelian algebra of symmetries
for $\mathcal{D}$. Nevertheless, for a given distribution $\mathcal{D}$, it is
difficult to find explicitly such a local chart.
\end{remark}

In the case of $1$-dimensional distributions, the main result of \cite{Ba} can
be stated as follows

\begin{proposition}
\label{Main_basarab}Let $\{Y_{1},...,Y_{n-1}\}$ be a solvable structure for a
$1$-dimensional distribution $\mathcal{D}=<Z>$ on an orientable $n$%
-dimensional manifold $N$. Then, for any given volume form $\Omega$ on $N$,
$\mathcal{D}$ can be described as the annihilator space of the system of
$1$-forms $\{\omega_{1},...,\omega_{n-1}\}$ defined as
\begin{equation}
\omega_{i}=\frac{1}{\Delta}(Y_{1}\lrcorner...\lrcorner\widehat{Y}_{i}%
\lrcorner...\lrcorner Y_{n-1}\lrcorner Z\lrcorner\Omega), \label{forme_omega}%
\end{equation}
where the hat denotes omission of the corresponding field, $\lrcorner$ denotes
the insertion operator (see \cite{Spivak}) and $\Delta$ is the function
defined by
\[
\Delta=Y_{1}\lrcorner Y_{2}\lrcorner...\lrcorner Y_{n-1}\lrcorner
Z\lrcorner\Omega.
\]
Moreover, one has that
\[%
\begin{array}
[c]{l}%
d\omega_{n-1}=0,\\
d\omega_{i}=0\hspace{0.1in}mod\{\omega_{i+1},...,\omega_{n-1}\}
\end{array}
\]
for any $i\in\{1,...,n-2\}$.
\end{proposition}
\begin{proof}
Here we only give a sketch of the proof. The reader is referred to
\cite{Ba} or \cite{BaPri} for more details. In view of the
definition of the forms
$\{\omega_{i}\}$ one readily gets that%
\begin{equation}
\mathcal{D}_{s}=Ann(\omega_{s+1},...,\omega_{n-1}),\qquad\forall
s=0,1,...,n-2.\label{D_s_E}%
\end{equation}
All these distributions, for $s=0,...,n-1$, are completely
integrable for the definition of solvable structure. Hence, by
applying Frobenius theorem to the distribution
$\mathcal{D}_{n-2}=Ann(\omega_{n-1})$, one should get
$d\omega_{n-1}=\alpha\wedge\omega_{n-1}$ for some $1$-form
$\alpha$. On the other hand, by using algebraic formula for the
exterior derivative $d$ and the formula defining the function
$\Delta$, one can show that $1/\Delta$ is an
integrating factor for the $1$-form $Y_{1}\lrcorner...\lrcorner Y_{n-2}%
\lrcorner Z\lrcorner\Omega$. Hence, one gets that
$d\omega_{n-1}=0$. The rest of the proof is based on similar
computations which allow one to show that, on the integral
manifolds of $\{\omega_{i+1},...,\omega_{n-1}\}$, one has that
$1/\Delta$ is still an integrating factor for the $1$-forms $Y_{1}%
\lrcorner...\lrcorner\widehat{Y}_{i}\lrcorner...\lrcorner
Y_{n-1}\lrcorner Z\lrcorner\Omega$.
\end{proof}

It follows that, under the assumptions of this proposition, $\mathcal{D}$ can
be completely integrated by quadratures (at least locally). In fact, in view
of $d\omega_{n-1}=0$, locally $\omega_{n-1}=dI_{n-1}$ for some smooth function
$I_{n-1}$. Now, since $d\omega_{n-2}=0$ mod$\{\omega_{n-1}\}$, $\omega_{n-2}$
is closed on the level manifolds of $I_{n-1}$. Then, iterating this procedure,
it is possible to compute all the integrals $\{I_{1},...,I_{n-1}\}$ of
$\mathcal{D}$ and eventually find its (local) integral manifolds in implicit
form $\{I_{1}=c_{1},...,I_{n-1}=c_{n-1}\}$.

\begin{remark}
\label{Riduz_parz}If one knows a solvable structure, Proposition
\ref{Main_basarab} allows a complete reduction (or integration) of
$\mathcal{D}$. However, if $Y_{1}$ is a symmetry of $\mathcal{D}=<Z>$ and
$Z,Y_{1},Y_{2},...,Y_{n-1}$ is just a frame on $M$, then one can still define
the forms $\{\omega_{i}\}$ and $\mathcal{D}=Ann(\omega_{1},...,\omega_{n-1})$.
In this case, one cannot completely integrate $\mathcal{D}$, but just reduce
its codimension by one. This partial reduction is a particular case of that
described in the paper \cite{Anderson-Fels} and in terms of the forms
$\omega_{i}$'s can be described as follows. Since $Y_{1}$ is a symmetry of
$\mathcal{D}$, then it is also a symmetry of the Pfaffian system $\mathcal{I}$
generated by $\omega_{1},...,\omega_{n-1}$. It follows that locally the
quotient of $M$ by the action induced by $Y_{1}$ is a manifold $\overline{M}$
naturally equipped with a Pfaffian system $\overline{\mathcal{I}}$. Since
$Y_{1}\lrcorner\omega_{j}=\delta_{1j}$, one has that $\overline{\mathcal{I}%
}=\{\omega_{2},...,\omega_{n-1}\}$. Now the above mentioned reduced
distribution is just $\overline{\mathcal{D}}=Ann(\overline{\mathcal{I}})$.
\end{remark}

Proposition \ref{Main_basarab} can also be applied to the integration of ODEs.
In fact, in view of above discussion, one can think of an ODE $\mathcal{E}$ as
a manifold equipped with the $1$-dimensional distribution $\mathcal{D}%
=<\overline{D}_{x}>$. In particular, equation $\mathcal{E}$ defined by
(\ref{ODE_scalare}) can be seen as a $(k+1)$-dimensional submanifold of
$J^{k}(\pi)$ naturally equipped with the volume form
\begin{equation}
\Omega=dx\wedge du\wedge...\wedge du_{k-1}, \label{volume_E}%
\end{equation}
 where $\wedge$ denotes the usual wedge product.

Then, by applying Proposition \ref{Main_basarab} to ODEs one readily gets the following

\begin{proposition}
\label{Main_ODE}Let $\{Y_{1},...,Y_{k}\}$ be a solvable structure for the
$1$-dimensional distribution $\mathcal{D}=<\overline{D}_{x}>$ on $\mathcal{E}$
defined by (\ref{ODE_scalare}). Then $\mathcal{E}$ is integrable by
quadratures and the general solution of $\mathcal{E}$ can be obtained in
implicit form by subsequently integrating the system of one forms $\omega
_{k},...,\omega_{1}$, in the given order.
\end{proposition}

\begin{remark}
The facts observed above in Remark \ref{Riduz_parz} apply also to the case of
ODEs. As an example, consider the ODE $\mathcal{E}:=\{u_{2}=-u/x^{2}\}$. This
ODE admits the classical symmetry $x\partial_{x}$. The reduced equation
$\overline{\mathcal{E}}$ is equipped with the Pfaffian system $\overline
{\mathcal{I}}$ generated by $\omega_{2}$, or equivalently by $\Delta\omega
_{2}$. Since $\Delta\omega_{2}$ is a $1$-form on the quotient $\overline
{\mathcal{E}}$, then it can be written only in terms of the invariants of
$Y_{1}$. Now, since the invariants of $Y_{1}$ are just $I_{1}=u,I_{2}=xu_{x}$,
one can check that $\overline{\mathcal{I}}=\{dI_{2}-(I_{2}-I_{1})dI_{1}%
/I_{2}\}$. Then integral manifolds of $\overline{\mathcal{I}}$ are exactly the
solutions of the reduced ODE $dI_{2}/dI_{1}=(I_{2}-I_{1})/I_{2}$ obtained by
using standard symmetry reduction method \cite{Olv1}.
\end{remark}

\section{New applications of solvable structures to the integration of ODEs}

In this section we discuss our new reduction scheme for ODEs.

As remarked above, solvable structures have been introduced to generalize the
standard symmetry reduction method. Finding solvable structures, however, is
not always easy and in practice a noteworthy simplification may come by
computing solvable structures adapted to an admitted symmetry algebra
$\mathcal{G}$, if any. In fact, as discussed below, under this assumption the
determining equation of a solvable structure extending $\mathcal{G}$ become
more affordable. Moreover, since the language of covering allows one to handle
nonlocal symmetries just like local symmetries, we can take advantage of
solvable structures adapted also to nonlocal symmetries. In this paper we consider 
solvable structures adapted to any kind (local, higher and nonlocal)
of symmetry algebras.

Now we first discuss the practical determination of solvable structures for
ODEs. To this end, it is useful to have the following technical Lemma which
expresses symmetry condition in the framework of multivector fields (see
\cite{Chern}).

\begin{lemma}
\label{lemma1}Let $\mathcal{D}=<X_{1},...,X_{r}>$ be an $r$-dimensional
distribution on an $n$-dimensional manifold $N$. The vector field $X$ is a
symmetry of $\mathcal{D}$ iff
\begin{equation}
L_{X}(X_{i})\wedge X_{1}\wedge X_{2}\wedge...\wedge X_{r}=0,\qquad\forall
i=1,2,...,r. \label{*}%
\end{equation}

\end{lemma}

In view of this Lemma, in fact, one readily gets the following

\begin{proposition}
\label{proposition}Let $\mathcal{E}$ be defined by (\ref{ODE_scalare}),
$\Omega$ be the volume form (\ref{volume_E}) and $\omega_{i}$'s defined by
(\ref{forme_omega}). The vector fields $\{Y_{1},...,Y_{k}\}$ on $\mathcal{E}$
determine a solvable structure for $\mathcal{D}=<\overline{D}_{x}>$ iff the
following two conditions are satisfied:\newline(a) $Y_{1}\lrcorner
Y_{2}\lrcorner...\lrcorner Y_{k}\lrcorner\overline{D}_{x}\lrcorner\Omega\neq
0$;\newline(b) for any $s=0,1,...,k-1$, defining $Y_{0}=\overline{D}_{x}$, one
has%
\begin{equation}
L_{Y_{s+1}}(Y_{j})\wedge Y_{0}\wedge Y_{1}\wedge...\wedge Y_{s}=0,\qquad
\forall j=0,1,...,s. \label{Deter_eq_solv_str}%
\end{equation}

\end{proposition}

\begin{proof}
(a) By using the notations of Definition \ref{Def_solv}, one has that
$\mathcal{D}_{k}=T\mathcal{E}$. Hence the vector fields $Y_{1},...,Y_{k}%
,\overline{D}_{x}$ are linearly independent and the conclusion
follows.\newline(b) The result follows by Proposition \ref{Main_basarab} and
(\ref{*}).
\end{proof}


\begin{remark}
By using the volume form $\Omega$, conditions (\ref{Deter_eq_solv_str}) can
also be rewritten in the alternative form%
\begin{equation}
L_{Y_{s+1}}(Y_{j})\lrcorner\left(  Y_{0}\lrcorner Y_{1}\lrcorner...\lrcorner
Y_{s}\lrcorner\Omega\right)  =0,\qquad\forall j=0,1,...,s.
\label{deter_solv_str_*}%
\end{equation}

\end{remark}

Equations (\ref{Deter_eq_solv_str}) (or (\ref{deter_solv_str_*})) are the
determining equations for the solvable structures $\{Y_{1},...,Y_{k}\}$ of
$\mathcal{D}=<\overline{D}_{x}>$. If one writes down all these equations in
coordinate form, it turns out that these form a system of
\[
\sum_{s=0}^{k-1}(s+1)
{k+1 \choose s+2}%
 =1-2^{k+1}+(k+1)2^{k}%
\]
nonlinear first order PDEs in the $(k+1)k$ unknown components of the fields
$Y_{1},...,Y_{k}$. It is underdetermined when $k=1,2$ and overdetermined when
$k\geq3$. However, this system has a very special form. In fact, for $s=0$
equation (\ref{Deter_eq_solv_str}) gives a system $S_{0}$ of equations
involving only the field $Y_{1}$. Then, after determining a solution $Y_{1}$
and substituting in (\ref{Deter_eq_solv_str}), for $s=1$ one gets another
system $S_{1}$ of equations only for $Y_{2}$. It follows that, iterating this
procedure, the determination of solvable structures just follows by the
subsequent analysis of the linear systems $S_{0},...,S_{k-1}$.

\begin{remark}
In view of the definition of solvable structure, any solution $Y_{1}$ of
$S_{0}$ is a symmetry of the vector field $\overline{D}_{x}$ and hence it is a
symmetry of $\mathcal{E}\subset J^{k}(\pi)$. Nevertheless, it does not mean
that it is a symmetry of the contact distribution $\mathcal{C}^{k}$ on
$J^{k}(\pi)$. In fact $Y_{1}$ could also be an internal symmetry of
$\mathcal{E}$. For the details on possible relations between external,
internal and generalized symmetries the reader is referred to \cite{AKO}.
\end{remark}

In view of Remark \ref{Remark}, the determining equations for solvable
structures always admit local solutions even though they could be hardly
solvable in practice. However, a noteworthy simplification may occur when one
only knows $\{Y_{1},...,Y_{r}\}$, $r\leq k-1$, and a complete system of joint
invariants $\{\gamma_{1},...,\gamma_{k-r}\}$ for the distribution
$\mathcal{D}_{r}$. In fact, using the same notations of Definition
\ref{Def_solv} and Proposition \ref{Main_basarab} one has the following

\begin{proposition}
\label{prop_invar_cong}Let $\mathcal{D}=<Z>$ be a $1$-dimensional distribution
and $\{Y_{1},...,Y_{r}\}$, $r\leq k-1$, be such that $L_{Y_{h}}\mathcal{D}%
_{h-1}\subset\mathcal{D}_{h-1}$ for any $h\in\{1,...,r\}$. If one knows a
complete system of joint invariants $\{\gamma_{1},...,\gamma_{k-r-1}\}$ for
the distribution $\mathcal{D}_{r}$, then $\mathcal{D}$ is completely
integrable by quadratures.
\end{proposition}

\begin{proof}
Let $\{\gamma_{i},g_{j}\}$ be a local coordinate chart adapted to the leaves
of the distribution $\mathcal{D}_{r}=Ann\{d\gamma_{1},...,d\gamma_{k-r-1}\}$.
It turns out that any field $\partial_{\gamma_{i}}$ is a symmetry of
$\mathcal{D}_{r}$, for $L_{\partial_{\gamma_{i}}}(d\gamma_{j})=0$, and hence
$\{Y_{1},...,Y_{r},\partial_{\gamma_{1}},...,\partial_{\gamma_{k-r-1}}\}$ form
a solvable structure. The conclusion follows by a direct application of
Proposition \ref{Main_basarab}.
\end{proof}

It is clear that, since the joint invariants $F\in C^{\infty}(\mathcal{E})$
for $\mathcal{D}_{r}$ are described by the overdetermined system%
\[
\overline{D}_{x}(F)=Y_{i}(F)=0,\qquad\forall i=1,...,r,
\]
the more fields $Y_{i}$ one knows, the more it is feasible to find a
fundamental set of invariants $\{\gamma_{1},...,\gamma_{k-r}\}$. In fact, the
analysis of differential consequences of such an overdetermined system of PDEs
often lead to the determination of the general solution.

Above facts, in particular, apply to the case when $\mathcal{E}$ has an
$r$-dimensional solvable algebra\footnote{We denote by $<Y_{1},...,Y_{r}%
>_{\mathbb{R}}$the $\mathbb{R}$-linear span of the fields $Y_{1},...,Y_{r}$.}
$\mathcal{G}=<Y_{1},...,Y_{r}>_{\mathbb{R}}$ of local classical symmetries. In
such cases, the problem of finding solvable structures which extends
$\mathcal{G}$ is much easier and feasible if $r$ is sufficiently large. A
first simple implementation of this idea, even though with a different
approach, has been already given in the paper \cite{HaAt}. In fact, in that
paper it is shown how such a kind of completion together with the results of
\cite{Ba} may lead to the integration of a second order ODE whose Lie algebra
of point symmetries is $1$-dimensional.

However there is no real need to limit ourselves to classical local
symmetries. In fact, in view of above discussion, the determination of
solvable structures can be noteworthy simplified if one preliminarily enlarges
$\mathcal{G}$ by adjoining as many symmetries as possible. Moreover, as
already observed in the previous sections, one often runs into ODEs with a
lack of local symmetries. Hence, in such cases, to overcome the occurring
difficulties in the reduction procedure one could also take advantage of the
existence of nonlocal symmetries.

However, as discussed above, nonlocal symmetries of $\mathcal{E}$ can only be
taken into account by preliminarily embedding $\mathcal{E}$ in an auxiliary
system $\mathcal{E}^{\prime}$, of the form (\ref{covering}), and then
searching for symmetries (higher or classical) of $\mathcal{E}^{\prime}$. As a
consequence, the nonlocal symmetry reduction of $\mathcal{E}$ comes from the
local symmetry reduction of $\mathcal{E}^{\prime}$. By applying Proposition
\ref{Main_basarab} above to the distribution $<\widetilde{D}_{x}>$ on
$\mathcal{E}^{\prime}$, one readily gets the following analogue of Proposition
\ref{Main_ODE}:

\begin{proposition}
Let $\{Y_{1},...,Y_{k+1}\}$ be a solvable structure for the $1$-dimensional
distribution $\widetilde{\mathcal{D}}=<\widetilde{D}_{x}>$ on the covering
system $\mathcal{E}^{\prime}$ of a $k$-order ODE $\mathcal{E}$. Then the
distribution $\widetilde{\mathcal{D}}$ is integrable by quadratures and the
general solution of $\mathcal{E}^{\prime}$ can be obtained in implicit form by
subsequently integrating the system of one forms $\omega_{k+1},...,\omega_{1}%
$, in the given order.
\end{proposition}

Moreover, since any covering system $\mathcal{E}^{\prime}$ of the form
(\ref{covering}) is naturally equipped with the volume form
\[
\widetilde{\Omega}=dx\wedge du\wedge...\wedge du_{k-1}\wedge dw,
\]
one also has the following analogue of Proposition \ref{proposition}:

\begin{proposition}
The vector fields $\{Y_{1},...,Y_{k+1}\}$ on $\mathcal{E}^{\prime}$ determine
a solvable structure for $\widetilde{\mathcal{D}}=<\widetilde{D}_{x}>$ iff the
following two conditions are satisfied:\newline(a) $Y_{1}\lrcorner
Y_{2}\lrcorner...\lrcorner Y_{k+1}\lrcorner\widetilde{D}_{x}\lrcorner
\widetilde{\Omega}\neq0$;\newline(b) for any $s=0,1,...,k$, defining
$Y_{0}=\widetilde{D}_{x}$, one has%
\[
L_{Y_{s+1}}(Y_{j})\wedge Y_{0}\wedge Y_{1}\wedge...\wedge Y_{s}=0,\qquad
\forall j=0,1,...,s.
\]

\end{proposition}

Unfortunately, if one considers a general covering system $\mathcal{E}%
^{\prime}$, it is not true in general that $\mathcal{E}^{\prime}$ inherits the
local symmetries of $\mathcal{E}$. Hence, if for example $\mathcal{E}$ has a
solvable algebra $\mathcal{G}$ of local symmetries, by considering a generic
covering $\mathcal{E}^{\prime}$ one not only may lose the local symmetries
belonging to $\mathcal{G}$ but also acquires one more dimension. It follows
that, in this case, one should also determine a solvable structure with one
more vector field.

In particular, when working with solvable symmetry algebras, above facts raise
the question of weather $\mathcal{E}^{\prime}$ can be chosen in such a way
that $\mathcal{G}$ is inherited and the dimensional growth is compensated by
the presence of some new nonlocal symmetry. An answer to this question is
provided by the following

\begin{proposition}
Let $\mathcal{E}$ be a $k$-order ODE of the form (\ref{ODE_scalare}) which
admits an algebra $\mathcal{G}$ of local symmetries. The algebra $\mathcal{G}$
is inherited by a covering system $\mathcal{E}^{\prime}$, of the form
(\ref{covering}), iff $H$ is a joint scalar invariant of $\mathcal{G}$. In
particular, if $\partial_{w}H=0$, $\mathcal{G}$ can be extended to the Lie
algebra $\widehat{\mathcal{G}}=\mathcal{G}\oplus\left\langle \partial
_{w}\right\rangle _{\mathbb{R}}$. If in addition $\mathcal{G}$ is solvable,
then also $\widehat{\mathcal{G}}$ is solvable.
\end{proposition}

\begin{proof}
Each $X\in\mathcal{G}$ is a vector field tangent to $\mathcal{E}$ of the form
$X=\overline{D}_{x}^{i}(\psi)\partial_{u_{i}}$. The proof of the first part
readily follows by imposing the invariance of $\mathcal{E}^{\prime}$under the
vector fields of $\mathcal{G}$. In fact, since each $X\in\mathcal{G}$ is such
that $\left.  X(u_{k}-f)\right|  _{\mathcal{E}^{(\infty)}}=0$, one readily
gets that $X(w_{1}-H)=-X(H)=0$. Hence, if $H$ is a function of the joint
scalar invariants of $\mathcal{G}$ and $\partial_{w}H=0$, one gains also the
new symmetry $\partial_{w}$. The thesis follows by observing that
$\partial_{w}$ commute with all the vector fields of $\mathcal{G}$.
\end{proof}

However, when $\partial_{w}H=0$, the advantage of solvable structures on
$\mathcal{E}^{\prime}$ is that one does not necessarily need that the symmetry
algebra $\mathcal{G}$ is inherited by $\mathcal{E}^{\prime}$. In fact, since
any $X\in\mathcal{G}$ is such that%
\[
\lbrack X,\overline{D}_{x}]=\alpha\overline{D}_{x}%
\]
for some $\alpha$, then one readily gets that $X$ is a symmetry of the
distribution $<\widetilde{D}_{x},\partial_{w}>$. Hence, if $\mathcal{G}%
=<Y_{1},Y_{2},...,Y_{r}>_{\mathbb{R}}$, one could search just for solvable
structures for $<\widetilde{D}_{x}>$ which extend $\{\partial_{w},Y_{1}%
,Y_{2},...,Y_{r}\}$ and possibly include nonlocal symmetries.

In particular all above facts apply to the special case when $\mathcal{E}%
^{\prime}$ is a $\lambda$-covering system for $\mathcal{E}$. In this case, if
$Y$ is a nonlocal symmetry corresponding to a $\lambda$-symmetry, one could
search for solvable structures which extend for example the $2$-dimensional
algebra spanned by $\partial_{w}$ and $Y$.

\section{Examples}

The following examples provide an insight into the applications of above
reduction scheme. All the examples make use of solvable structures adapted to
symmetries of the given ODE. In particular, since in a solvable structure
there is no need to distinguish between symmetries and other vector fields,
all these vector fields will be treated on the same footing. \bigskip
\newline\textbf{Example 1.} Consider the ODE
\begin{equation}
u_{3}=\frac{u_{2}^{2}(u_{1}^{2}-2uu_{2})}{u_{1}^{4}}. \label{ex1}%
\end{equation}
The algebra of classical symmetries $\mathcal{G}_{0}$ of (\ref{ex1}) is
$2$-dimensional and generated by the functions
\[%
\begin{array}
[c]{l}%
\varphi_{1}=xu_{1}-u,\\
\varphi_{2}=u_{1}.
\end{array}
\]
Therefore classical symmetries do not suffice for the complete symmetry
reduction of (\ref{ex1}).

\begin{remark}
By reducing equation (\ref{ex1}) through the symmetries $X_{2},X_{1}$ generated by
$\varphi_{2}$ and $\varphi_{1}$, in the given order, one first obtains the
reduced equation $U_{\xi\xi}=(-U^{-3}+U^{-2})U_{\xi}^{2}-2\xi U^{-3}U_{\xi
}^{3}$ and subsequently the Riccati equation $V_{\eta}=1-2\eta^{-3}V^{2}%
-\eta^{-1}V+\eta^{-2}V$. Here by $U=u_{1},\xi=u$ we denote the basis
invariants of $X_{2}$ and analogously by $V=\xi U_{\xi}$ and $\eta=U$ the
basis invariants of $X_{1}$. Above Riccati equation can be explicitly
integrated even though it is not possible to exhibit an explicit solution of
the initial equation (\ref{ex1}). Below we show an alternative integration of
(\ref{ex1}) by means of a solvable structure.
\end{remark}

However, one can also search for higher symmetries. In this case, one can find
that up to first order the algebra of symmetries is $\mathcal{G}_{0}%
\oplus\mathcal{G}_{1}$ where $\mathcal{G}_{1}$ is a $3$-dimensional algebra
generated by the functions
\[%
\begin{array}
[c]{l}%
\varphi_{3}=xu_{1}^{2}-2uu_{1},\\
\varphi_{4}=u_{1}^{2},\\
\varphi_{5}=e^{1/u_{1}}u_{1}^{2}.
\end{array}
\]
Let us denote by $X_{i}$ the symmetries generated by the functions
$\varphi_{i}$ above. We have that the only nonvanishing commutators are
\[%
\begin{array}
[c]{l}%
\lbrack X_{1},X_{2}]=X_{2},\qquad\lbrack X_{1},X_{4}]=X_{4},\qquad\lbrack
X_{1},X_{5}]=X_{5},\\
\lbrack X_{2},X_{3}]=-X_{4},\qquad\lbrack X_{3},X_{5}]=X_{5}.
\end{array}
\]
Therefore, $X_{1},X_{2}$ and $X_{4}$ span a solvable algebra which can be used
to completely reduce (\ref{ex1}). This special example shows the advantage of
considering also higher order symmetries in the reduction of ODE.\newline Here
we will apply the solvable algebra $\{Y_{1}=X_{4},Y_{2}=X_{2},Y_{3}=X_{1}\}$
in the integration of (\ref{ex1}) through the subsequent integration of the
forms $\omega_{3},\omega_{2}$ and $\omega_{1}$.\newline Now, by making use of
(\ref{forme_omega}) and coordinate expressions of the fields $Y_{i}$ and $Z=\overline{D}_{x}$, one gets
$\Delta=(u_{1}^{2}-2uu_{2})u_{2}^{2}$ and%
\[%
\begin{array}
[c]{l}%
\omega_{3}=\dfrac{2u_{2}}{u_{1}^{2}-2uu_{2}}du+\dfrac{2uu_{2}-u_{1}^{2}%
-2u_{1}^{3}}{(u_{1}^{2}-2uu_{2})u_{1}^{2}}du_{1}\vspace{0.08in}\\
\hspace{0.55in}+\dfrac{u_{1}^{2}}{(u_{1}^{2}-2uu_{2})u_{2}}du_{2}%
,\vspace{0.08in}\\
\omega_{2}=dx+\dfrac{2xu_{2}-2u_{1}}{u_{1}^{2}-2uu_{2}}du+\dfrac
{(xu_{1}-2u)u_{1}}{(u_{1}^{2}-2uu_{2})u_{2}}du_{2}\vspace{0.08in}\\
\hspace{0.55in}+\dfrac{(2xuu_{1}-4u^{2})u_{2}^{2}+(-2xu_{1}^{4}+(2u-x)u_{1}%
^{3}+2uu_{1}^{2})u_{2}+u_{1}^{5}}{(u_{1}^{2}-2uu_{2})u_{2}u_{1}^{3}}%
du_{1},\vspace{0.08in}\\
\omega_{1}=-\dfrac{1}{u_{1}^{2}-2uu_{2}}du+\dfrac{u_{1}^{5}+uu_{1}^{2}%
u_{2}-2u^{2}u_{2}^{2}}{(u_{1}^{2}-2uu_{2})u_{1}^{4}u_{2}}du_{1}\vspace
{0.08in}\\
\hspace{0.55in}-\dfrac{u}{(u_{1}^{2}-2uu_{2})u_{2}}du_{2}.
\end{array}
\]
The $1$-form $\omega_{3}$ is exact and one can check that $\omega_{3}=dI_{3}$
with%
\[
I_{3}=\frac{1}{u_{1}}+\ln\left(  \frac{u_{2}}{2uu_{2}-u_{1}^{2}}\right)  .
\]
Then, by restricting $\omega_{2}$ and $\omega_{1}$ on the integral manifolds
$\Sigma_{3}:=\{I_{3}=c_{3}\}$, one finds%
\[%
\begin{array}
[c]{l}%
\left.  \omega_{2}\right\vert _{\Sigma_{3}}=dx+\dfrac{e^{\frac{1+u_{1}%
\ln(u_{2})-c_{3}u_{1}}{u_{1}}}-u_{1}^{3}}{u_{1}^{3}u_{2}}du_{1}+\dfrac{u_{1}%
}{u_{2}^{2}}du_{2},\vspace{0.08in}\\
\left.  \omega_{1}\right\vert _{\Sigma_{3}}=\dfrac{e^{\frac{1+u_{1}\ln
(u_{2})-c_{3}u_{1}}{u_{1}}}}{2u_{1}^{4}u_{2}}du_{1}+\dfrac{1}{2u_{2}^{2}%
}du_{2}.
\end{array}
\]
In this case both $\left.  \omega_{2}\right\vert _{\Sigma_{3}}$ and $\left.
\omega_{1}\right\vert _{\Sigma_{3}}$ are exact $1$-forms. This is due to the
fact that both $\{Y_{1},Y_{2},Y_{3}\}$ and $\{Y_{2},Y_{1},Y_{3}\}$ are
solvable structures. Then, when restricted on the level manifold $\Sigma_{3}$,
$\omega_{2}$ and $\omega_{1}$ are the exterior derivative of the functions
\[%
\begin{array}
[c]{l}%
I_{2}=x-\frac{u_{1}}{u_{2}}+\frac{(u_{1}-1)e^{\frac{1-c_{3}u_{1}}{u_{1}}}%
}{u_{1}},\vspace{0.08in}\\
I_{1}=-\frac{1}{2u_{2}}-\frac{1}{2}\frac{(1-2u_{1}+2u_{1}^{2})e^{\frac
{1-c_{3}u_{1}}{u_{1}}}}{u_{1}^{2}},
\end{array}
\]
respectively. It follows that the general solution of (\ref{ex1}) can be
written in the implicit form $\{I_{1}=c_{1},I_{2}=c_{2},I_{3}=c_{3}%
\}$.\bigskip\newline\textbf{Example 2.} Consider the ODE
\begin{equation}
u_{2}=-\frac{x^{2}}{4u^{3}}-u-\frac{1}{2u}.\label{ex2}%
\end{equation}
As shown in \cite{MuRo}, (\ref{ex2}) has no point symmetries but admits a
$\lambda$-symmetry with $\lambda=x/u^{2}$. If we consider the system
\begin{equation}
\left\{
\begin{array}
[c]{l}%
u_{2}=-\dfrac{x^{2}}{4u^{3}}-u-\dfrac{1}{2u},\\
w_{1}=\dfrac{x}{u^{2}}%
\end{array}
\right.  \label{ex_2_cov}%
\end{equation}
a nonlocal symmetry $Y$ of (\ref{ex2}) which corresponds to the $\lambda
$-symmetry found by Muriel and Romero in \cite{MuRo} is that generated by the
functions $\varphi^{1}=ue^{w}$ and $\varphi^{2}=-2e^{w}$. We will search for
solvable structures which extend the nonabelian algebra $\mathcal{G}%
=<Y_{1}=Y,Y_{2}=\partial_{w}>$.\newline By making use of the Maple 11
routines, it is not difficult to find the most general solvable structure
$\{Y_{1},Y_{2},Y_{3}\}$ for $\mathcal{D}=\left\langle Y_{0}\right\rangle $
(recall that, in this case, $Y_{0}$ is the total derivative operator
restricted to (\ref{ex_2_cov})). The determining equations
(\ref{Deter_eq_solv_str}), in fact, now are $L_{Y_{3}}(Y_{j})\wedge
Y_{0}\wedge Y_{1}\wedge Y_{2}=0$, $j=0,1,2$, and admit the general solution
$Y_{3}=a_{1}\partial_{x}+a_{2}\partial_{u}+a_{3}\partial_{u_{1}}+a_{4}%
\partial_{w}$, with%
\[
a_{1}=\frac{\left(  4u^{4}+4u^{2}u_{1}^{2}+4xuu_{1}+x^{2}\right)  F+8u\left(
\left(  uu_{1}+x\right)  a_{2}-u^{2}a_{3}\right)  }{8u^{4}+u^{2}(8u_{1}%
^{2}+4)+8xuu_{1}+2x^{2}},
\]
$a_{2},a_{3},a_{4}$ arbitrary functions of $(x,u,u_{1},w)$ and $F$ an
arbitrary function of $x+\arctan\left(  (2uu_{1}+x)/(2u^{2})\right)
$.\newline Nevertheless, since here we are only interested in the integration
of (\ref{ex2}), we just consider the following particular solution (with
$F=2$, $a_{2}=a_{4}=0$ and $a_{3}=-1/(2u)$)%
\[
Y_{3}=\partial_{x}-\frac{1}{2u}\partial_{u_{1}}.
\]
In this case, since both $\{Y_{1},Y_{2},Y_{3}\}$ and $\{Y_{1},Y_{3},Y_{2}\}$
are solvable structures, one can check that $d\omega_{3}=d\omega_{2}=0$.
Hence, by proceeding like in the previous example, one finds $\omega
_{3}=dI_{3}$, $\omega_{2}=dI_{2}$ with
\[%
\begin{array}
[c]{l}%
I_{2}=2\ln u-w-\ln(4u^{2}u_{1}^{2}+4xuu_{1}+4u^{4}+x^{2}),\\
I_{3}=\arctan\left(  \dfrac{u^{2}u_{1}+xu/2}{u^{3}}\right)  +x.
\end{array}
\]
It follows that, when restricted on the level manifold $\{I_{2}=c_{2}%
,I_{3}=c_{3}\}$, $\omega_{1}$ is an exact $1$-form and one can check that it
is the exterior derivative of the function%
\[
I_{1}=\frac{2e^{c_{2}}u^{2}}{\cos(x-c_{3})^{2}}+2e^{c_{2}}\ln(\cos
(c_{3}-x))-2e^{c_{2}}x\tan(c_{3}-x).
\]
Then, one gets the general solution of (\ref{ex_2_cov}) in the implicit form
$\{I_{1}=c_{1},I_{2}=c_{2},I_{3}=c_{3}\}$. In this case, however, one can
solve with respect to $u$ and gets the solution of (\ref{ex2}) in the form%
\[
u=\pm\cos(x-c_{2})\sqrt{\frac{c_{1}e^{-c_{3}}}{2}-x\tan(x-c_{2})-\ln
(\cos(x-c_{2}))}%
\]
This solution depends on $3$ arbitrary constants, but of course one of them is
inessential. In fact, by suitably rearranging $c_{1}$ one can write the
general solution of (\ref{ex2}) in the form
\[
u=\pm\cos(x-C_{2})\sqrt{C_{1}-x\tan(x-C_{2})-\ln(\cos(x-C_{2}))},\qquad
C_{1},C_{2}\in\mathbb{R}.
\]
This solution coincides with that found in \cite{MuRo}. \bigskip
\newline\textbf{Example 3.} Consider the ODE
\begin{equation}
u_{2}=\frac{u_{1}}{(1+x)x}-\frac{x^{2}+x^{3}}{4(1+x)u^{3}}-\frac{x}%
{2(1+x)u}.\label{ex3}%
\end{equation}
It can be shown that (\ref{ex3}) has no point symmetries but admits a
$\lambda$-covering with $\lambda=x/u^{2}$. Hence, if we consider the covering
system $\mathcal{E}^{\prime}$ defined by (\ref{ex3}) together with
$w_{1}=\lambda$, it admits the nonlocal symmetry $Y$ generated by $\varphi
^{1}=ue^{w}$ and $\varphi^{2}=-2e^{w}$. We will search for solvable structures
which extend the nonabelian algebra $\mathcal{G}=<Y_{1}=Y,Y_{2}=\partial_{w}%
>$.\newline Contrary to the previous example, in this case it is more
difficult to find the most general solution of the determining equations
(\ref{Deter_eq_solv_str}). Nevertheless, one can find some particular
solutions. In this case we just consider the following particular solution%
\[
Y_{3}=\frac{(1+x)(x+2uu_{1})^{2}}{8u_{1}(x+uu_{1})}\partial_{u}+\frac
{x+x^{2}+2uu_{1}+2xuu_{1}}{4u^{2}(x+uu_{1})}\partial_{w}.
\]
As in the previous example, since both $\{Y_{1},Y_{2},Y_{3}\}$ and
$\{Y_{1},Y_{3},Y_{2}\}$ are solvable structures, one can check that
$d\omega_{3}=d\omega_{2}=0$. Then, by proceeding as in the previous examples,
one gets the general solution of (\ref{ex3}):
\begin{equation}
u=\pm(C-2x+2\ln(1+x))\left(  \int\frac{-x}{\left(  C-2x+2\ln(1+x)\right)
^{2}}dx\right)  ^{1/2},\label{sol_ex_3}%
\end{equation}
with $C\in\mathbb{R}$. Notice that, as in the previous example, in the general
integral of the covering system $\mathcal{E}^{\prime}$ there was a third
inessential constant we have gauged out in (\ref{sol_ex_3}). Moreover, we have
absorbed one of the arbitrary constants in the indefinite integral.
\bigskip\newline\textbf{Example 4.} Consider the ODE
\begin{equation}
u_{2}=(xu_{1}-xu^{2}+u^{2})e^{-1/u}+\frac{2u_{1}^{2}}{u}+u_{1}.\label{ex4}%
\end{equation}
It can be shown that (\ref{ex4}) has no point symmetries but admits a
$\lambda$-covering with $\lambda=xe^{-1/u}-1/x$. Hence, if we consider the
covering system $\mathcal{E}^{\prime}$ defined by (\ref{ex4}) together with
$w_{1}=\lambda$, it admits the nonlocal symmetry $Y$ generated by $\varphi
^{1}=xe^{w}$ and $\varphi^{2}=xe^{w}$. We will search for solvable structures
which extend the nonabelian algebra $\mathcal{G}=<Y_{1}=Y,Y_{2}=\partial_{w}%
>$.\newline Also in this case, it is very complicated to find the most general
solution of the determining equations (\ref{Deter_eq_solv_str}). Nevertheless,
one can again find some particular solutions. In this case we just consider
the following particular solution%
\[
Y_{3}=u^{2}e^{x}\partial_{u_{1}}.
\]
Contrary to the previous examples, $\{Y_{1},Y_{3},Y_{2}\}$ is not a solvable
structure, hence in this case only $\omega_{3}$ is closed. By proceeding as in
the previous examples, one still can find the general solution of the covering
system $\mathcal{E}^{\prime}$. Then, by solving with respect to $u$ one gets
the general solution of (\ref{ex4}):
\begin{equation}
u=\frac{1}{Ce^{x}+\ln\left(
{\displaystyle\int}
\dfrac{-x}{e^{Ce^{x}}}dx\right)  },\label{sol_ex_4}%
\end{equation}
with $C\in\mathbb{R}$. Here, as in the previous examples, we have gauged out
from (\ref{sol_ex_4}) all inessential arbitrary constants. \bigskip
\newline\textbf{Example 5.} Consider the ODE
\begin{equation}
u_{2}=\frac{2u_{1}^{2}}{u}+\left(  xe^{x/u}-\frac{4}{x}\right)  u_{1}-\left(
\frac{3u^{2}}{x}+u\right)  e^{x/u}+xu^{2}+\frac{2u}{x^{2}}.\label{ex5}%
\end{equation}
It can be shown that (\ref{ex5}) has no point symmetries but admits a
$\lambda$-covering with $\lambda=xe^{x/u}$. Hence, if we consider the covering
system $\mathcal{E}^{\prime}$ defined by (\ref{ex5}) together with
$w_{1}=\lambda$, it admits the nonlocal symmetry $Y$ generated by $\varphi
^{1}=e^{w}u^{2}/x$ and $\varphi^{2}=-e^{w}$. We will search for solvable
structures which extend the nonabelian algebra $\mathcal{G}=<Y_{1}%
=Y,Y_{2}=\partial_{w}>$.\newline Also in this case, it is very complicated to
find the most general solution of the determining equations
(\ref{Deter_eq_solv_str}). Nevertheless, one can again find some particular
solutions. In this case we just consider the following particular solution%
\[
Y_{3}=\frac{u^{2}}{x^{3}}\partial_{u_{1}}.
\]
By proceeding as in the previous examples, one still can find the general
solution of the covering system $\mathcal{E}^{\prime}$. Then, by solving with
respect to $u$ one gets the general solution of (\ref{ex5}):
\begin{equation}
u=\frac{-x^{2}}{C+\dfrac{x^{5}}{20}+x\ln\left(
{\displaystyle\int}
\dfrac{-x}{e^{\frac{C}{x}+\frac{x^{4}}{20}}}dx\right)  },\label{sol_ex_5}%
\end{equation}
with $C\in\mathbb{R}$. Here, as in the previous examples, we have gauged out
from (\ref{sol_ex_5}) all inessential arbitrary constants.

\section{Concluding remarks}

Solvable structures for a 1-dimensional distribution $\mathcal{D}$ can be determined by
solving the system of equations (\ref{Deter_eq_solv_str}). The analysis of
these determining equations is in general more
feasible if one makes use of symbolic manipulation packages like those written
for Maple V distributions. In particular, we suggest the packages
DifferentialGeometry and Jets by I. Anderson (and coworkers) and M. Marvan,
respectively. The first package is recommended for various computations
involved in the determination of solvable structures and the further
integration of the forms $\omega_{i}$. The second package is very helpful in
the computation of symmetries and the analysis of differential consequences of
a given system of equations.

\begin{acknowledgement}
The authors would like to thank G. Cicogna and G. Gaeta for constructive
discussions and suggestions. This work was partially supported by GNFM-INDAM
through the project \textquotedblright Simmetrie e riduzione per PDE, principi
di sovrapposizione e strutture nonlocali\textquotedblright.
\end{acknowledgement}



\begin{thebibliography}{99}                                                                                               %


\bibitem {Anderson-Fels}I. Anderson and M. Fels 2005 Exterior Differential
Systems with Symmetry \textit{Acta Appl. Math.} \textbf{87} 3-31

\bibitem {AKO}I. Anderson, N. Kamran and P. Olver 1993 Internal, External and
Generalized Symmetries \textit{Adv. Math. }\textbf{100} 53-100

\bibitem {Ba}P. Basarab-Horwath 1992 Integrability by quadratures for systems
of involutive vector fields \textit{Ukrain. Mat. Zh.} \textbf{43} 1330-1337;
translation in \textit{Ukrain. Math. J.} \textbf{43} (1992) 1236-1242

\bibitem {BaPri}M.A. Barco and G.E. Prince 2001 Solvable symmetry structures
in differential form applications \textit{Acta Appl. Math. } \textbf{66} 89-121

\bibitem {BaPri1}M.A. Barco and G.E. Prince 2001 New symmetry solution
techniques for first-order non-linear PDEs \textit{Appl. Math. Comput.
}\textbf{124} 169-196

\bibitem {Bluman-Anco}G. W. Bluman and S. C. Anco 2002 \textit{Symmetry and
Integration Methods for Differential Equations} (Berlin: Springer)

\bibitem {Bluman-Kumei}G. W. Bluman and S. Kumei 1989 \textit{Symmetries and
differential equations} (Berlin: Springer)

\bibitem {Bluman-Reid}G. W. Bluman and G. J. Reid 1988 New symmetries for
ordinary differential equations \textit{IMA J. Appl. Math.} \textbf{40} 87-94

\bibitem {Ca}D. Catalano Ferraioli 2007 Nonlocal aspects of $\lambda
$-symmetries and ODEs reduction \textit{J. Phys. A: Math Theor} \textbf{40} 5479-5489

\bibitem {Chern}S. S. Chern, W. H. Chen and K. S. Lam 2000 \textit{Lectures on
Differential Geometry} (Singapore: World Scientific)

\bibitem {CicGaeMor}Cicogna G, Gaeta G and Morando P 2004 On the relation
between standard and $\mu$-symmetries for PDEs \textit{J. Phys. A:Math. Gen.}
\textbf{37} 9467-86

\bibitem {Fe}Fels M E 2007 Integrating scalar ordinary differential equations
with symmetry revisited \textit{Found. Comput. Math.} \textbf{7} 417-454

\bibitem {Gae}G. Gaeta 1994 \textit{Nonlinear symmetries and nonlinear
equations} (Dordrecht: Kluwer)

\bibitem {GaMo}G. Gaeta and P. Morando 2004 On the geometry of
lambda-symmetries and PDEs reduction \textit{J. Phys. A:Math. Gen.}
\textbf{37} 6955-6975

\bibitem {GMM}M.L. Gandarias, E. Medina and C. Muriel 2002 New symmetry
reductions for some ordinary differential equations \textit{J. Nonlin. Math.
Phys.} \textbf{9} Suppl.1 47-58

\bibitem {Govinder-Leach}K. S. Govinder and P. G. L. Leach 1997 A
group-theoretic approach to a class of second-order ordinary differential
equations not possessing Lie point symmetries \textit{J. Phys. A} \textbf{30} 2055-2068

\bibitem {HaAt}T. Hartl and C. Athorne 1994 Solvable structures and hidden
symmetries \textit{J. Phys. A: Math Gen} \textbf{27} 3463-3471

\bibitem {Kam}N. Kamran 2002 \textit{Selected topics in the geometrical study
of differential equations} (Providence, RI: American Mathematical Society)

\bibitem {Kras-Vin}I. S. Krasil'shchik and A. M. Vinogradov 1989 Nonlocal
Trends in the Geometry of Differential Equations: Symmetries, Conservation
Laws, and B\"{a}cklund Transformations \textit{Acta Appl. Math.} \textbf{15} 161-209

\bibitem {Mo}P. Morando 2007 Deformation of Lie derivative and $\mu
$-symmetries \textit{J. Phys. A: Math. Theor.} \textbf{40} 11547-11559

\bibitem {MuRo}C. Muriel and J.L. Romero 2001 New method of reduction for
ordinary differential equations \textit{IMA Journal of Applied mathematics}
\textbf{66} 111-125

\bibitem {MuOl}C. Muriel and J.L. Romero and P. Olver 2006 Variational
$C^{\infty}$-symmetries and Euler-Lagrange equations \textit{J. Differential
equations} \textbf{222} 164-184

\bibitem {Olv1}P.J. Olver 1993 \textit{Application of Lie groups to
differential equations} (Berlin: Springer)

\bibitem {Olv2}P.J. Olver 1995 \textit{Equivalence, invariants, and symmetry}
(Cambridge: Cambridge University Press)

\bibitem {Sac}E. Pucci and G. Saccomandi 2002 On the reduction methods for
ordinary differential equations \textit{J. Phys. A} \textbf{35} 6145-6155

\bibitem {ShePri}J. Sherring and G. Prince 1992 Geometric aspects of reduction
of order \textit{Trans. Amer. Math. Soc.} \textbf{334} 433-453

\bibitem {Spivak}M. Spivak 1999 \textit{A comprehensive introduction to
differential geometry, Vol 1 }(Houston, Texas: Publish or Perish, Inc.)

\bibitem {Ste}H. Stephani 1989 \textit{Differential equations. Their solution
using symmetries} (Cambridge: Cambridge University Press)

\bibitem {Vin-et-al}A. M. Vinogradov et al. 1999 \textit{Symmetries and
conservation laws for differential equations of mathematical physics}
(Providence, RI: American Mathematical Society)
\end{thebibliography}
\end{document}